\documentclass[10pt,twocolumn]{IEEEtran}
\IEEEoverridecommandlockouts                              % This command is only
\usepackage{IRM}

\title{\LARGE \bf
Control Contraction Metrics:  Convex and Intrinsic Criteria for Nonlinear Feedback Design
}
\author{
Ian R. Manchester \ \ Jean-Jacques E. Slotine
 \thanks{I. R. Manchester is with the Australian Centre for Field Robotics and  School of Aerospace, Mechanical and Mechatronic Engineering, University of Sydney, NSW 2006, Australia. Email: ian.manchester@sydney.edu.au.
J.-J. E. Slotine is with the Nonlinear Systems Laboratory, Massachusetts   Institute   of   Technology, Cambridge, MA, 02139. Email: jjs@mit.edu. This work was supported by the Australian Research Council (DP150100577).}
}

\begin{document}

\maketitle

\begin{abstract}
 We introduce the concept of a control contraction metric, extending contraction analysis to constructive nonlinear control design. We derive sufficient conditions for exponential stabilizability of all trajectories of a nonlinear control system. The conditions have a simple geometrical interpretation, can be written as a convex feasibility problem, and are invariant under coordinate changes. We show that these conditions are necessary and sufficient for feedback linearizable systems, and also  derive novel convex criteria for exponential stabilization of a nonlinear submanifold of state space. 
We illustrate the benefits of convexity by constructing a controller for an unstable polynomial system that combines local optimality and global stability, using a metric found via sum-of-squares programming.
\end{abstract}

\section{Introduction}
The concept of a Lyapunov function is central in nonlinear system analysis, and builds upon the intuitive notion of a system's energy dissipating over time \cite{khalil2002nonlinear, slotine1991applied}. For nonlinear control design, the natural extension is the control Lyapunov function (CLF): a generalized measure of energy that can be {\em made} to decrease by choice of control action, first formalized in \cite{artstein1983stabilization, sontag1983lyapunov} though implicit in earlier works.

If a CLF can be found, then surprisingly simple formulas yield stabilizing feedback designs for quite broad classes of systems \cite{sontag1989universal}, \cite[Ch. 4]{freeman2008robust}. However, the fundamental problem of {\em finding} a CLF remains challenging. For mechanical and electrical systems, physical energy often yields effective choices  \cite{slotine1991applied}, while for systems of particular ``triangular'' structures, backstepping and related methods can be applied \cite{khalil2002nonlinear, freeman2008robust}. 

For linear systems, a straightforward change of variables converts the CLF criteria to a linear matrix inequality (LMI) \cite{boyd1994linear}, \cite{dullerud2000course}, but for nonlinear systems the set of CLFs for a particular system is not necessarily convex or even connected \cite{rantzer2001dual}.  The density functions of \cite{rantzer2001dual}, the occupation measures of \cite{lasserre2008nonlinear}, and the Lyapunov measures of \cite{vaidya2010nonlinear} each offer ``dual'' representations that yield convex (but generally infinite-dimensional) searches guaranteeing almost-everywhere stability. Computationally tractable finite-dimensional approximations have been  based on gridding or the sum-of-squares relaxation \cite{parrilo2003semidefinite}.

The main result of this paper is that if the nonlinear stabilization problem is studied {\em differentially}, then the simple convexification results for linear systems are recovered by generalising the concept of a {\em contraction metric}.

Contraction analysis is based on the study of a nonlinear system by way of its differential dynamics (a.k.a. variational system) along solutions \cite{Lohmiller98}. Roughly speaking, since the differential dynamics are linear time-varying (LTV), many techniques from linear systems theory can be directly applied. A central result is that if all solutions of a smooth nonlinear system are locally exponentially stable in a common metric, then all solutions are globally exponentially stable. %A Riemannian contraction metric can be thought of as a family of quadratic Lyapunov functions for the LTV differential dynamics, and the Riemannian distance is then an incremental Lyapunov function for the original nonlinear system. 
Historically, basic convergence results on contracting systems can be traced back to the results of \cite{lewis1949metric} in terms of Finsler metrics, further explored in \cite{forni2014differential}, while convex conditions for existence and robustness of limit cycles were given in \cite{manchester2014transverse}. In contraction analysis, question of stability is decoupled from the specification of particular solutions, and this property is also relevant for control design: it is common in industrial plant-wide control and robot motion control to have a ``layered architecture'' in which a higher-layer generates target trajectories, and a lower layer guarantees accurate tracking. 

In this paper we introduce {\em universal stabilizability}: the property that every forward-complete solution of a system can be globally stabilized. We then define a {\em control contraction metric} (CCM) for a nonlinear system, and show that existence of a CCM is sufficient for universal exponential stabilizability. We also give extensions for stabilization of submanifolds.

The CCM stabilizability condition has a simple geometric interpretation: small displacements in directions orthogonal to the span of the control inputs must be ``naturally'' contracting. Here both the notions of orthogonality and contraction depend on the choice of metric. This can be thought of as a differential version of the CLF condition of \cite{artstein1983stabilization, sontag1989universal}. 
While the resulting conditions are arguably quite strong, in contrast to a CLF our criteria for the existence for a CCM can be formulated as a convex feasibility problem. Furthermore, unlike e.g. backstepping they are invariant under smooth changes of coordinates and affine feedback transformations, and necessary and sufficient for feedback linearizable systems.

The feedback controllers we propose will generally involve real-time optimization to find a minimal-length path with respect to the metric (a geodesic) joining the current state to the desired state. This problem is generally simpler than that in nonlinear model predictive control (MPC) since it has lower dimension, lacks dynamic constraints, and minimal geodesics are guaranteed to exist \cite{leung2017}.  If a state-independent metric exists, then geodesics are just straight lines and our method is closely related to well-known methods using quadratic Lyapunov functions, e.g. \cite[p. 99]{boyd1994linear} and  \cite{pavlov2006uniform}.

The main feedback controller we propose is smooth almost everywhere, but as with nonlinear MPC, our controller may be discontinuous at some points in state-space. To ensure existence of solutions at such points, we also propose a sampled-data controller that runs open-loop short intervals, a common strategy in nonlinear MPC \cite{fontes2003discontinuous}, and similar to the notion of s-stability introduced in \cite{clarke1997asymptotic}.

\section{Preliminaries and Problem Setup}\label{sec:setup}
For symmetric matrices $A,B$ the notation $A\ge B$ ($A>B$) means that $A-B$ is positive semidefinite (positive definite). 
%For a matrix $B$, the notation $B_\perp$ refers to an {\em annihilator}, i.e. a full-rank matrix satisfying $B'B_\perp=0$. 
The non-negative reals are denoted $\RR^+:=[0,\infty)$. Given a smooth matrix function $M(x,t)$ and vector field $v:\R^n\times\R^+\to \R^n$ defined for $x\in\R^n,t\in\R^+$, we use the following notation for directional derivative $\partial_v M := \sum_j\pder[M]{x_j}v_{j}$.  A set of sample times is a sequence $t_0, t_1, ...$ with $t_0=0$ and  $t_{i}<t_{i+1}$ for all $i$ and $t_i\rightarrow\infty$ as $i\rightarrow\infty$.

%$D^+$ denotes the upper Dini derivative: $D^+c(t) = \limsup_{h\searrow 0} (c(t+h)-c(t))/h$. 
%For simplicity of expression we generally assume functions are smooth ($C^\infty$), although in most cases this could be relaxed.

In this paper we consider control-affine nonlinear systems:
\begin{equation}\label{eq:sys}
\dot x = f(x,t)+B(x,t)u
\end{equation}
where $x(t)\in\RR^n, u(t)\in\RR^m$ are state and control, respectively, at time $t\in\RR^+$, and $f, B$ are smooth functions of their arguments. We denote the $i^{th}$ column of $B(x,t)$ by $b_i(x,t)$. We assume $u(t)$ is at least piecewise-continuous, and \eqref{eq:sys} holds with the right-derivative at points of discontinuity.

We define a {\em target trajectory} to be a forward-complete solution of \eqref{eq:sys}, i.e. a pair $(x^\star, u^\star)$ with $x^\star:\R^+\to \R^n$ piecewise differentiable and $u^\star:\R^+\to \R^m$ piecewise-continuous satisfying \eqref{eq:sys} for all $t\in \R^+$. 

We will consider open-loop, sampled-data and continuous feedback controllers. An open-loop controller is a mapping $(x(0),x^\star(\cdot), u^\star(\cdot),t)\mapsto u(t)$.  Given a set of sample times, a sampled-data feedback controller has the property that on each interval $[t_i,t_{i+1}),$ the control law is a mapping $(x(t_i),x^\star(t),u^\star(t),t)\mapsto u(t)$, while a continuous feedback controller is a mapping $(x(t),x^\star(t),u^\star(t), t)\mapsto u(t)$.

A target trajectory $x^\star, u^\star$ is said to be globally exponentially controllable (resp. stabilizable) if one can construct an open-loop (resp. feedback) controller such that for any initial condition $x(0)\in\R^n$, a unique solution $x(t)$ of \eqref{eq:sys} exists for all $t$ and satisfies
\begin{equation}\label{eq:expstab}
|x(t)-x^\star(t)|\le e^{-\lambda t}R|x(0)-x^\star(0)|,
\end{equation} 
where rate $\lambda>0$ and overshoot $R>0$ are constants independent of initial conditions. If every target trajectory is globally exponentially controllable (resp. stabilizable) then the system is said to be {\em universally exponentially controllable (resp. stabilizable)}.
The following example illustrates the distinction with global stabilizability of a particular solution.
\begin{ex} 
Consider the planar system
\[
\begin{bmatrix}\dot x_1 \\ \dot x_2\end{bmatrix} = \begin{bmatrix}-2x_1+x_1^2-x_2^2\\-6x_2+2x_1x_2\end{bmatrix}+\begin{bmatrix}1\\0\end{bmatrix}u,
\]
which has four equilibria \cite{rantzer2001dual}: $(x_1, x_2)=(0,0), \, (2, 0)$ and $(3,\pm \sqrt{3})$. The origin is globally exponentially stabilized by the feedback law
$
u = -x_1^2-x_2^2
$, since the quadratic Lyapunov function $V(x) = x_1^2+x_2^2$ verifies $\dot V = -4x_1^2-12x_2^2\le -4V$. However, notice that if $x_2(0)=0$ then $x_2(t)=0, \forall t\ge 0$ regardless of $x_1, u$, so it is impossible for a control input to move the state from the line $x_2=0$ to the equilibria at $(3,\pm \sqrt{3})$. Therefore this system is {\em not} universally stabilizable.
\end{ex}

We utilize the following standard results of Riemannian geometry, see, e.g., \cite{docarmo1992riemannian} for details. A Riemannian metric is a smoothly-varying inner product $\ip{\cdot}{\cdot}_x$ on the tangent space of a state manifold $\X$; this defines local notions of length, angle, and orthogonality. In this paper $\X=\R^n$ and the tangent space can also be identified with $\R^n$. We allow metrics to be smoothly time-varying, and use the following notation: $\ip{\delta_1}{\delta_2}_{x,t}=\delta_1'M(x,t)\delta_2$ and $\|\delta\|_{x,t} = \sqrt{\ip{\delta}{\delta}_{x,t}}$. We call a metric {\em uniformly bounded} if $\exists \alpha_2\ge \alpha_2>0$ such that $\alpha_1 I\le M(x,t) \le \alpha_2 I$ for all $x,t$. For a smooth curve $c:[0,1]\to\R^n$ we use the notation $c_s(s) :=\pder[c(s)]{s}$, and define the Riemannian length and energy functionals as 
\[
L(c,t):=\int_0^1\|c_s\|_{c,t}ds, \ \ E(c,t):=\int_0^1\|c_s\|^2_{c,t}ds,\] 
respectively, with integration interpreted as the summation of integrals for each smooth piece. 
Let $\Gamma$ be the set of piecewise-smooth curves $[0,1]\to\R^n$, and for a pair of points $x,y\in\R^n$, let $\Gamma(x,y)$ be the subset of $\Gamma$ connecting $x$ to $y$, i.e. curves $c\in\Gamma(x,y)$ if $c\in\Gamma$, $c(0)=x$ and $c(1)=y$. A smooth curve $c(s)$ is {\em regular} if $\pder[c]{s}\ne 0$ for all $s\in[0,1]$. The Riemannian distance $d(x,y,t):=\inf_{c\in\Gamma(x,y)}L(c,t)$, and we define $E(x,y,t):=d(x,y,t)^2$. Under the conditions of the Hopf-Rinow theorem a smooth, regular minimum-length curve (a geodesic) $\gamma$ exists connecting every such pair, and the energy and length satisfy the following inequalities: $E(x,y,t)=E(\gamma,t)=L(\gamma,t)^2\le L(c,t)^2 \le E(c,t)$ where $c$ is any curve joining $x$ and $y$. For time-varying paths $c(t,s)$, we also write $c(t):=c(t,\cdot):[0,1]\to\R^n$.

A central result of \cite{Lohmiller98} is that if there exists a uniformly bounded metric $M(x,t)$ such that
$
\dot M + \pder[f]{x}'M+M\pder[f]{x}\le -2\lambda M,
$
where $\dot M = \pder[M]{t}+\partial_f M$, 
then the system is contracting with rate $\lambda$, i.e. $\ddt \|\delta_x\|_{x,t}\le -\lambda \|\delta_x\|_{x,t}$.
By integrating along minimizing geodesics, we see that $d(x,y,t)$ and $E(x,y,t)$ between any pair of points $x,y$ both decrease exponentially under the flow of the system, and thus can serve as  incremental Lyapunov functions. Such systems are called {\em contracting} systems and $M$ is a {\em contraction metric}. Similarly, we call a system {\em strictly contracting with rate $\lambda$} if $\ddt \|\delta_x\|_{x,t}< -\lambda \|\delta_x\|_{x,t}$ for $\delta_x\ne 0$. Since $M(x,t)>0$, a system which is contracting with rate $\lambda>0$ is strictly contracting with any rate less than $\lambda$. %so the ``strictness'' could be thought of as an attribute of the {\em rate} more than the {\em system}.

%The proofs of our results are collected in the appendix.

%We restrict attention to Riemannian metrics in this paper, i.e. $V(x,\delta_x,t):=\d_x'M(x,t)\d_x$, because the resulting conditions can be made convex, however the basic strategy can easily be extended to, e.g., Finsler metrics \cite{bao2000introduction}, \cite{lewis1949metric}, \cite{forni2014differential}.

%
%%
%Variations on this theme can be constructed with other measures of differential length, including vector norms \cite{Lohmiller98} and Finsler metrics \cite{lewis1949metric,forni2014differential}.

%%% Local Variables: 
%%% mode: latex
%%% TeX-master: "CCM_Journal"
%%% End: 

\section{Control Contraction Metrics}\label{sec:ccm}
To analyse stabilizability we utilize the ``extended'' system consisting of \eqref{eq:sys} paired with its differential  dynamics:
\begin{equation}
\dot\delta_x = A(x,u,t)\delta_x+B(x,t)\delta_u,\label{eq:diffdyn}
\end{equation}
defined along solutions $x(t), u(t)$, where 
$
A := \pder[f]{x}+\sum_{i = 1}^m\pder[b_i]{x}u_i.
$
%Here $\delta_x(t)$ should be thought of as a tangent vector to a smooth path of states through $x(t)$, and $\delta_u(t)$ as a tangent vector to a smooth path of controls through $u(t)$.

Let us begin by examining the case when a known controller makes the system contracting, and yet is flexible enough that any target trajectory of \eqref{eq:sys} remain possible in closed-loop.

\begin{prop}\label{prop:contracting}
	Suppose there exists a smooth feedback control law $u= k(x,t)+v$ that makes the closed-loop system strictly contracting with rate $\lambda$ in some metric $M(x,t)$ for any piecewise-continuous signal $v(t)$. %Assume that by choice of $v$ all feasible trajectories of \eqref{eq:sys} remain feasible. I.e. for every $u_0\in\R^m, x_0\in\R^n,t_0\in \R^+$ there exists a $v_0\in\R^m$ such that $ k(x_0,v_0,t_0) = u_0$. 
	Then for all $x,u,t$
\begin{equation}\label{eq:ccm_explicit}
\dot M +(A+BK)'M+M(A+BK)<-2\lambda M,	
\end{equation}
where $K= \pder[k]{x}$, and for $\delta_x\ne 0$ the following is true:
\begin{equation}\label{eq:CCM}
\d_x'MB=0 \implies \d_x'(\dot M + A'M+MA+2\lambda M)\d_x<0.
\end{equation}
\end{prop}
It is clear that \eqref{eq:ccm_explicit} $\Longrightarrow$ \eqref{eq:CCM}; the proof of \eqref{eq:ccm_explicit}, and all subsequent results in this paper, are collected in the appendix.
%
%Let us begin by discussing the case of a known feedback control law $u= k(x,v,t)$, where $v\in\R^m$ is an auxiliary input, that makes the closed-loop system contracting with rate $\lambda$ in some metric $M(x,t)$. Assume that by choice of $v$ all feasible trajectories of \eqref{eq:sys} remain feasible. I.e. for every $u_0\in\R^m, x_0\in\R^n,t_0\in \R^+$ there exists a $v_0\in\R^m$ such that $ k(x_0,v_0,t_0) = u_0$. It is straightforward to show that this implies
%\begin{equation}\label{eq:ccm_explicit}
%\dot M+(A+B K)'M+M(A+B K)\le -2\lambda M,
%\end{equation}
%for all $x,u,t$, where $ K = \pder[ k]{x}$ and $\dot M$ is the total derivative $\pder[M]{t}+\partial_{f+B u}M$. An immediate consequence of \eqref{eq:ccm_explicit} is that
%\begin{equation}\label{eq:CCM_le}
%\d_x'MB=0 \implies \d_x'(\dot M + A'M+MA+2\lambda M)\d_x\le 0.
%\end{equation}
%Taking square roots and writing in coordinate-free notation gives the statement that if $\ip{\d_x}{b_i(x)}_{x,t}=0 $  for all $i=1,2,..., m$ then $\ddt \|\d_x\|_{x,t}\le -\lambda \|\d_x\|_{x,t}.
%$
%I.e. {\em every tangent vector $\d_x$ orthogonal to the span of actuated directions $b_i(x)$ is naturally contracting with rate $\lambda$}.

Since \eqref{eq:CCM} is independent of the particular control law $k$, it describes an intrinsic property of the system \eqref{eq:sys}: if $\ip{\d_x}{b_i(x)}_{x,t}=0 $  for all $i=1,2,..., m$ then $\ddt \|\d_x\|_{x,t}< -\lambda \|\d_x\|_{x,t}.
$
I.e. {\em every tangent vector $\d_x$ orthogonal to the span of actuated directions $b_i(x)$ is naturally contracting with rate $\lambda$}.
 It is interesting to ask whether \eqref{eq:CCM} {\em implies the existence} of some form of stabilizing control for any target trajectory. Our main theoretical result is that this is indeed the case: %The first main result of the paper is the following:

\begin{thm} \label{thm:CCM} 
%Consider the system \eqref{eq:sys}, \eqref{eq:diffdyn}. 
If there exists a uniformly-bounded metric $M(x,t)$, i.e. $\alpha_1 I\le M(x,t) \le \alpha_2 I$, 
%$M(x,t)$, with bounds $\alpha_2\ge \alpha_1>0$, 
for which \eqref{eq:CCM} holds for all $\d_x\ne 0, x, u, t$, then System \eqref{eq:sys} is
\begin{enumerate}[ 1)]
\item universally exponentially open-loop controllable,
\item universally exponentially stabilizable via sampled-data feedback with arbitrary sample times,
\item universally exponentially stabilizable via continuous feedback defined almost everywhere, and everywhere in a neighbourhood of the target trajectory.
\end{enumerate}
all with rate $\lambda$ and overshoot $R=\sqrt{\tfrac{\alpha_1}{\alpha_2}}$.
\end{thm}
We refer to a metric satisfying the conditions of this theorem as a {\em control contraction metric} (CCM) for the system \eqref{eq:sys}. The proof of this theorem is given in the appendix, but here we briefly describe the main idea and construction of controllers.

Given a CCM, Lemma \ref{lem:lVdecr} in the appendix establishes the existence of a {\em differential feedback controller} $\delta_u=k_\delta(x,\delta_x,u,t)$ that achieves closed-loop exponential stabilization of the differential dynamics \eqref{eq:diffdyn} along all solutions:
\[
\ddt (\d_x'M\d_x)=\d_x'\dot M\d_x + 2\d_x'M(A\d_x+Bk_\delta)<-2\lambda \d_x'M\d_x
\]
and furthermore, is path-integrable, so that for any smooth path $c\in\Gamma$ and any $u^\star\in\R^m$ and $t\in\R^+$, the following integral equation has a unique solution:
\begin{equation}\label{eq:k_p}
k_p(c,u^\star,t,s):=u^\star+\int_0^s k_\delta (c(\s),c_s(\s), k_p(c,u^\star,t,\s), t)d\s.
\end{equation}
The motivation for this construction is to give a smooth path of control signals $k_p$ with tangent vectors $\pder[k_p]{s} = k_\delta$ for all $s\in[0,1]$ and boundary condition $k_p=u^\star$ at $s=0$.

The path-integrability condition is a significantly weaker requirement than the $k_\delta$ being {\em completely integrable}, i.e. of the form $k_\delta=K(x,t)\delta_x$, with $K$ the Jacobian of a feedback controller, as was assumed in Proposition \ref{prop:contracting}. This distinction will be important for our convex conditions in Section \ref{sec:dual}.

\subsubsection{Open-Loop Control}\label{sec:OL_construction}

For $t_i\ge 0$, consider a bounded or unbounded time interval in one of the following forms: $\mathcal T = [t_i, t_{i+1})$, $\mathcal T = [t_i, t_{i+1}]$ or $\mathcal T=[t_i,\infty)$. 
\begin{enumerate}
	\item At the initial time $t_i$ measure $x(t_i)$ and construct a smooth path $c(t_i)\in\Gamma(x^\star(t_i),x(t_i))$.
	\item For each $t\in\mathcal T$, apply the control signal $u(t) = k_p(c(t),u^\star(t),t,1)$, where $c(t)$ is the forward image of the path $c(t_i)$ with the path of controls \eqref{eq:k_p}, i.e. for each $s\in[0,1]$ and $t\in\mathcal T$, $c(t,s)$ satisfies	\begin{equation}\label{eq:forward_image}
\ddt c(t,s) = f(c(t,s),t)+B(c(t,s),t)k_p(c(t),u^\star(t),t,s).\end{equation}
\end{enumerate}
When this strategy is applied on an interval $\mathcal T = [t_i,\infty)$ the length of the curve $c(t)$ shrinks exponentially, and allows us to establish Theorem \ref{thm:CCM} claim 1.

The bound $R=\sqrt{\frac{\alpha_2}{\alpha_1}}$ given in the theorem is achieved if the initial path $c_0$ is a minimal geodesic joining $x^\star(t_i)$ to $x(t_i)$, existence of which is established in Lemma \ref{lem:hopf_rinow} in the appendix.
With any other initial path, exponential stability is still achieved with the same rate but perhaps with larger overshoot. Note that when $c(t_i)$ is a geodesic it is in general {\em not} the case that $c(t)$ is a geodesic for $t>t_i$, see Fig. \ref{fig:Gamma_x}.

%While this controller is well-defined, a practical implementation will generally require discretized numerical approximation of \eqref{eq:k_p} and \eqref{eq:forward_image}.

\begin{figure}
\begin{center}
\includegraphics[width=0.65\columnwidth]{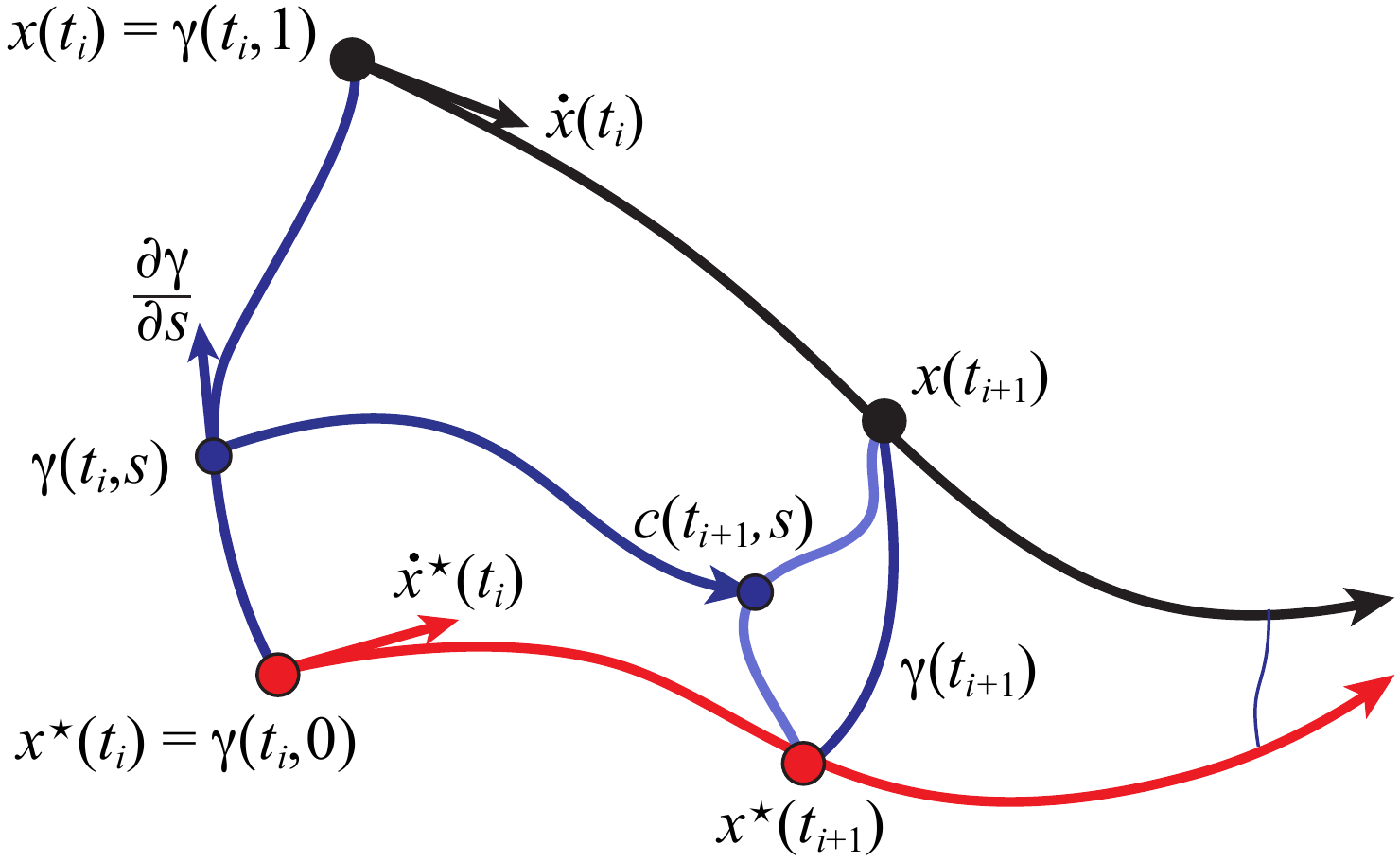}
\caption{An illustration of the geometry of solutions using the open-loop or sampled-data CCM-based control over an interval $[t_i, t_{i+1}]$. The target trajectory $x^\star(t)$ is shown in red, system trajectory $x(t)$ in black. Paths joining $x^\star(t)$ to $x(t)$ are shown in blue. }
\label{fig:Gamma_x}
\end{center}
\end{figure}

\subsubsection{Sampled-Data Feedback Controller}
The open-loop controller can be extended to sampled-data feedback by recomputing geodesics at the sampling instants. To be precise:
\begin{enumerate} 
\item At each sample time $t_i$, measure the state $x(t_i)$ and compute a minimal geodesic $\gamma_i \in \arg\min_{c\in\Gamma_i} E(c,t_i)$ where $\Gamma_i:=\Gamma(x^\star(t_i),x(t_i))$. 
\item On the interval $\mathcal T = [t_i, t_{i+1})$ apply the open-loop control described above with $c(t_i)=\gamma_i$.
\end{enumerate} 
Note that this is stabilizing with {\em any} choice of sample times,  including uniform sampling: $t_i=it_s$ for some fixed $t_s>0$.
%Lemma \ref{lem:hopf_rinow} in the appendix guarantees the existence of at least one minimal geodesic joining any pair of states when the metric is uniformly bounded.  %The proof is based on the fact that on each interval $[t_i, t_{i+1})$ the length of the curve $c(t)$ shrinks exponentially, then at $t=t_{i+1}$ it is replaced by an {\em even shorter} curve $\gamma_{i+1}$ and the process repeats.

\subsubsection{Smooth Feedback, Uniquely Defined Almost Everywhere and in a Neighbourhood of $x^\star$}

By taking the limit as sampling interval goes to zero, one can obtain a continuous-time controller which does away with the need to solve \eqref{eq:forward_image} over the inter-sample intervals. %The one difficulty is that the controller defined above may be discontinuous and even multiply-defined because of non-uniqueness of the minimizing geodesic, which can occur on a set zero Lebesgue measure. 
Specifically:

\begin{enumerate}
	\item Measure the state $x(t)$ and a minimal geodesic $\gamma = \arg\min_{c\in\Gamma(x^\star(t),x(t))} E(c,t_i)$.
\item Apply the control signal $u(t) =  k_p(\gamma,u^\star(t),t,1)$.
\end{enumerate}
This defines a mapping $(x(t), x^\star(t),u^\star(t),t\mapsto u(t)$, however a difficulty is that it may be multiply-defined or non-smooth at some states $x(t)$, specifically points on the {\em cut locus}, denoted by $\mathfrak C(x^\star,t)$, which is the set of points for which non-unique minimizing geodesics exist from $x^\star$ (cut points) and/or the first-order minimality condition fails (conjugate points). 

This set is known to have zero Lebesgue measure. Let use define $\mathfrak D(x^\star,t):=\R^n/(\mathfrak C(x^\star,t) \cup x^\star )$ which is diffeomorphic to punctured open ball. For every $x\in \mathfrak D(x^\star,t)$ there is a unique minimal geodesic $\gamma$ joining $x$ and $x^\star$ \cite[Ch 13]{docarmo1992riemannian}, and we show in the appendix that the above controller is smooth on $\mathfrak D(x^\star,t)$ and continuous at $x=x^\star(t)$. This controller is universally exponentially stabilizing under the technical assumption that the set of times at which $x(t)\in \mathfrak C(x^\star,t)$ has zero measure.

\subsection{Stronger Conditions Giving Simpler Controllers}\label{sec:strong}

Since the differential dynamics are linear, it is tempting to look for an admissible differential feedback controller of the form $\delta_u= K(x,t)\d_x$ satisfying \eqref{eq:ccm_explicit}. We will show that this is possible under the following slightly stronger conditions:

\begin{enumerate}[C1:]
\item \label{C1} if $\delta_x\ne 0$ satisfies $\delta_x'MB=0$, then
\[
\delta_x'\left(\pder[M]{t}+\partial_f M+\pder[f]{x}'M+M\pder[f]{x}+2\lambda M\right)\delta_x <0,
\]
\item \label{C2} for each  $i=1, 2, ..., m$,
$
\partial_{b_i} M + \pder[b_i]{x}'M+M\pder[b_i]{x} = 0.
$
\end{enumerate}

These stronger conditions also hold under the assumptions of Proposition \ref{prop:contracting}, as is clear from the proof in the appendix. Condition C\ref{C1} says that the uncontrolled system is contracting in directions orthogonal to the span of the control inputs. Condition C\ref{C2} ensures that large $u$ of unknown sign cannot cause expansion of $\|\d_x\|$.  Formally it states that the vector fields $b_i$ are Killing fields for the metric $M$. 

In particular, if $B$ is of the form $[0, I]'$, with $0$ and $I$ the zero and identity matrices of appropriate dimension, then Condition C\ref{C2} says that $M$ must not depend on the last $m$ state variables.  
By applying Finsler's theorem (see, e.g., \cite{uhlig1979recurring}) pointwise in $x$ and $t$ to condition C\ref{C1}, we immediately obtain the following:
\begin{prop} \label{prop:rho} Condition C\ref{C1} is equivalent to the existence of a scalar multiplier $\rho(x,t)$ such that for all $x,t$:
\begin{equation}\label{eq:CCM_rhoform}
\pder[M]{t}+\partial_f M+\pder[f]{x}'M+M\pder[f]{x}-\rho MBB'M+2\lambda M  < 0.
\end{equation}
\end{prop}
%\begin{proof}
%A version of Finsler's theorem (see, e.g. \cite{uhlig1979recurring}) states the following: given a square matrix $H$ and a matrix $G$, then $\delta_x'H\delta_x<0$ for all $\delta_x\ne 0$ satisfying $\delta_x'G=0$ if and only if there exists a scalar multiplier $\rho$ such that $H-\rho GG'<0$. Applying this results pointwise in $x$ and $t$ shows that Condition C\ref{C1} is equivalent to \eqref{eq:CCM_rhoform} with strict inequality, and therefore sufficient for non-strict inequality.
%\end{proof}
One can then construct the differential feedback gain $K(x,t) = -\frac{1}{2}\rho(x,t)B(x,t)'M(x,t)$ which satisfies \eqref{eq:ccm_explicit} and is always path integrable since it is independent of $u$.

\begin{remark}\label{rem:rho_gm}
If \eqref{eq:CCM_rhoform} holds for some multiplier $\rho(x,t)=\rho_0(x,t)$, then it clearly holds for any $\rho(x,t)\ge \rho_0(x,t)$ for all $x,t$ since $MBB'M\ge 0$. This can be interpreted as the differential feedback having infinite up-side gain margin, and also implies that one can construct a smooth $\rho(x,t)$.
\end{remark}

%%% Local Variables: 
%%% mode: latex
%%% TeX-master: "CCM_Journal"
%%% End: 

\subsection{Dual Metrics and Convexity of Synthesis}\label{sec:dual}

It is known that the search for a CLF for a linear system is convexified by a simple change of variables, leading to an LMI representation of stabilizability \cite{boyd1994linear, dullerud2000course}.  In this section we show that essentially the same transformation makes the search for a CCM convex.

Consider the change of variables $\eta = M(x,t)\delta_x$ and $W(x,t)=M(x,t)^{-1}$. This is related to the ``musical isomorphism" to the dual space of cotangent vectors, and the function $\eta'W\eta$ is the Fenchel dual of $\delta_x'M\delta_x$, so we refer to $W$ as a {\em dual CCM}.
 Under this change of variables, the CCM condition \eqref{eq:CCM} now states that $\eta'(-\dot W + AW +WA'+2\lambda W)\eta < 0$ whenever  $\eta'B=0$, which can be written as
 \begin{equation}
\label{eq:weak_kernel}
B_\perp'\left(-\dot W + AW +WA'+2\lambda W\right)B_\perp< 0,
\end{equation}
for all $x, u, t$, where $B_\perp'(x,t)$ is any matrix function satisfying $B_\perp'B=0$ for all $x,t$. Since differentiation is a linear operation, the inequality \eqref{eq:weak_kernel}  linear (and hence convex) in the unknown matrix function $W$.

One can search directly for differential feedback $\delta_u = K(x,u,t)\d_x$ by way of $W$ and $Y(x,u,t)\in R^{m\times n}$ satisfying
\begin{equation}\label{eq:CCM_Yform}
-\dot W + \pder[f]{x}W +W\pder[f]{x}' + BY + Y'B'+2\lambda W< 0,
\end{equation}
giving the differential feedback gain  $K = YW^{-1}$. If $Y$, and hence $K$, are at most affine in $u$ then the resulting differential control will be path-integrable, similarly to Lemma \ref{lem:lVdecr}.

Condition C\ref{C1} can be written similarly to \eqref{eq:weak_kernel}, and by Finsler's theorem is equivalent to the existence of a scalar function $\rho(x,t)$ satisfying the inequality
\begin{equation}\label{eq:CCM_rhoform_dual}
-\pder[W]{t}-\partial_f W + \pder[f]{x}W +W\pder[f]{x}' -\rho BB'+2\lambda W< 0.
\end{equation}
which is jointly convex in $W$ and $\rho$, and gives an explicit construction of a differential feedback gain $K = -\frac{1}{2}\rho B'W^{-1}$.

%A third, and more explicit, form is to search for $W$ and $Y(x,t)\in R^{m\times n}$ satisfying
%\begin{equation}\label{eq:CCM_Yform}
%-\dot W + \pder[f]{x}W +W\pder[f]{x}' + BY + Y'B'+2\lambda W\le 0,
%\end{equation}
%giving the differential feedback gain  $K = YW^{-1}$. 
%If written with strict inequalities than all three of \eqref{eq:CCM_kernelform}, \eqref{eq:CCM_rhoform_dual}, and \eqref{eq:CCM_Yform} are equivalent: \eqref{eq:CCM_kernelform}$\Leftrightarrow$\eqref{eq:CCM_rhoform_dual} is from Finsler's theorem, \eqref{eq:CCM_rhoform_dual}$\Rightarrow$\eqref{eq:CCM_Yform} via $Y = -\frac{1}{2}\rho B'$, \eqref{eq:CCM_Yform} $\Rightarrow$\eqref{eq:CCM_kernelform} since $B_\perp'BY=0$.
Condition C\ref{C2} also transforms to a linear constraint on $W$:
$\partial_{b_i}W - \pder[b_i]{x}W-W\pder[b_i]{x}'=0.
$

The above conditions are all convex but infinite-dimensional: they are inequalities that must hold over all $x\in\R^n$ and $t\in\R^+$, and the decision variables are sets of smooth matrix functions. Finite-dimensional LMI approximations can be constructed by building $W$ and $\rho$ or $Y$ as linear combinations of a finite basis set (e.g. polynomials up to some order), and verifying the inequalities either by gridding over states and times, or by the sum-of-squares relaxation \cite{parrilo2003semidefinite}.

\begin{remark}\label{rem:integrability}
Note that {\em complete integrability} of $k_\delta(x,\delta_x,t)=K(x,t)\delta_x$ could be imposed by requiring that each row of $K$ satisfies the Schwarz condition, i.e. $\pder[K_{i,j}]{x_k}=\pder[K_{i,k}]{x_j}$. While this constraint is linear and hence convex in $K$, it is {\em not} convex jointly in the decision variables $W, Y$ for \eqref{eq:CCM_Yform}, since $K=YW^{-1}$, or the decision variables $W, \rho$ for \eqref{eq:CCM_rhoform_dual}, since $K = -\frac{1}{2}\rho B'W^{-1}$. This is essentially the same problem as the well-known non-convexity of structured feedback synthesis for linear systems e.g. static output feedback \cite{syrmos1997static}.
\end{remark}

\section{Properties of Control Contraction Metrics}\label{sec:properties}
\subsection{Riemannian Energy as a CLF}

The proof of Theorem \ref{thm:CCM} uses an explicit construction of a particular stabilizing controller, but in doing so we have actually shown that the Riemannian energy $E(x, x^\star,t)$ can always be decreased, and hence be used as control Lyapunov function (CLF) for {\em any} target trajectory of the system. 

The formula for first variation of energy \cite[p. 195]{docarmo1992riemannian} gives a particularly convenient expression for the time derivative of the energy functional as an affine function of $u$:
\begin{align}
\frac{1}{2}\frac{d}{dt}E(x, x^\star, t) =& \langle \gamma_s(t,0), \dot x^\star \rangle_{x^\star,t} - \langle \gamma_s(t,1), f(x,t)\rangle_{x,t}\notag\\
&-\langle \gamma_s(t,1),B(x,t)u\rangle_{x,t}+\frac{1}{2}\pder[E]{t},\label{eq:Edot}
\end{align}
When $x(t)\in \mathfrak C(x^\star,t)$ the above formula still holds with $=$ replaced by $\le$ and $\frac{d}{dt}$ replaced by the Dini derivative. %The control contraction metric condition then implies the Artstein-Sontag CLF condition that if $\gamma_s(t,1)'M(x,t)B(x,t)=0$ then $\langle \gamma_s(t,0), \dot x^\star \rangle_{x^\star,t} - \langle \gamma_s(t,1), f(x,t)\rangle_{x,t}+\frac{1}{2}\pder[E]{t}< -\lambda E(x, x^\star,t)$.

In proving Theorem 1, we have also proven that for any $x^\star, u^\star, t$, the convex set (either a half-space or all of $\R^m$):
\[
\mathcal U =\left\{u\in\ R^m : \ddt E(x, x^\star, t) \le - 2\lambda E(x,x^\star, t) \right\},
%&+\langle \gamma_s(t,0), \dot x^\star \rangle_{x^\star,t} - \langle \gamma_s(t,1), f(x,t)\rangle_{x,t}+\frac{1}{2}\pder[E]{t}
\]
%This also makes precise the intuitive notion that the controller should push the state towards $x^\star$, where ``towards'' is defined by the direction in which a minimal geodesic departs from $x$. 
where $ \ddt E(x, x^\star, t)$ is given by \eqref{eq:Edot},  is always non-empty.

This opens up the possibility of using many other particular controllers based on CLFs that may have further desirable properties. For example, pointwise min-norm control \cite{freeman2008robust}:
$
u(t) = \arg\min_{\tilde u\in \mathcal U} \|\tilde u\|^2
$
would have reduced control magnitude, and can be generalized to provide approximate optimality with guaranteed stability \cite{primbs2000receding}.

\subsection{Invariance Under Coordinate Change and Feedback}
Metrics and dual metrics  are tensors: geometrical objects that are `intrinsic'' and have coordinate representations that transform appropriately under smooth coordinate changes. In the following theorem we establish that the CCM criteria are invariant under such coordinate changes and, additionally,  under affine feedback laws. 

\begin{thm}\label{thm:coord}
If the CCM condition \eqref{eq:CCM} (or equivalently \eqref{eq:weak_kernel})  is satisfied for system \eqref{eq:sys}, then \eqref{eq:CCM} and \eqref{eq:weak_kernel} still hold under:
\begin{enumerate}
\item affine feedback transformations $u(x,v) = \alpha(x)+\beta(x)v$ with $\beta$ a smooth non-singular $n\times n$ matrix function.
\item differential coordinate changes $\delta_\xi = \Phi(x)\delta_x$, in which $\Phi(x)$ is a non-singular matrix for all $x$, with the new CCM 
$
	M_\xi(x,t) :=\Psi'(x) M(x,t)\Psi(x)$ and dual CCM $
	W_\xi(x,t):= \Phi(x)  W(x,t) \Phi(x)'
$, where $\Psi(x) = \Phi(x)^{-1}$;
\item coordinate changes	 $\xi = \phi(x)$, $\phi$ a smooth diffeomorphism, with the new CCM and dual CCM $M_\xi, W_\xi$ as above with $\Phi(x) = \pder[\phi]{x}$ evaluated at $x=\phi^{-1}(\xi)$.
\end{enumerate}
\end{thm}

\begin{remark} If $\Phi(x)$ has bounded singular values over all $x,t$ within $[\sigma_{\min }, \sigma_{\max }]$, then the uniform bounds on $M$ and $W$ are also preserved under coordinate change.
\end{remark}
\subsection{Necessity for Feedback Linearizable Systems}
A corollary of Theorem \ref{thm:coord} is that for feedback linearizable systems, existence of a CCM is guaranteed. A system of the form \eqref{eq:sys}
is feedback linearizable if there exists a change of variables and feedback transformation such that the transformed system is linear time-invariant:
$
\dot \xi = G\xi+Hv,
$
where the pair of constant matrices $(G, H)$ is controllable \cite{khalil2002nonlinear}. 

%When such a transformation can be found a feedback stabilizer can be designed using LTI methods, but a major challenge is that even if such functions can be proven to exist via Frobenius theorem, they may be difficult to construct explicitly: this generally involves solving a partial differential equation. 

\begin{cor} For any feedback linearizable system there is a control contraction metric that verifies universal stabilizability, given by $W(x,t) = \Phi(x,t)P\Phi(x,t)'$ where $P$ is any constant symmetric positive definite matrix satisfying
$
H_\perp(GP+PG')H_\perp'<0.
$
\end{cor}
The proof is immediate from Theorem \ref{thm:coord}, and we note that such a $P$ is guaranteed to exist if $(G, H)$ is stabilizable \cite{dullerud2000course}.

The converse is not true: the necessary and sufficient conditions for feedback-linearizability consist of a controllability condition and an involutivity (complete integrability) condition to find an appropriate coordinate change \cite{khalil2002nonlinear}. In contrast, existence of a CCM depends only on stabilizability, and does not require complete integrability of the differential control.  This is because a metric corresponds to a {\em differential} change of coordinates $\delta_z = \Theta(x,t)\delta_x$, i.e. $M = \Theta'\Theta$, but there is no requirement that this should be integrable to an explicit change of coordinates $z=\theta(x)$. %The price paid is that implementation may require on-line optimization to find a geodesic.

\begin{ex}
The following system
\[
\ddt \bm{x_1\\x_2} = \bm{-x_1-x_1^3 +x_2^2\\0}+\bm{0\\1}u=:f(x)+Bu
\]
is not feedback linearizable in the sense of \cite{isidori1995nonlinear}, since the vector fields $B$ and 
$ ad_fB := \pder[f]{x}B-\pder[B]{x}f =[2x_2, 0]' 
$
are not linearly independent when $x_2=0$. However, it is universally stabilizable as verified by \eqref{eq:CCM_rhoform_dual} with $W=I$ and multiplier $\rho(x) =1+2x_2^2$. Additionally, since we can take $B_\perp=[1 \ 0]'$, condition \eqref{eq:weak_kernel} reduces to the fact that $\pder[f_1]{x_2}=-1-3x_1^2<0 \ \forall x$.
\end{ex}
%

%%% Local Variables: 
%%% mode: latex
%%% TeX-master: "CCM_Journal"
%%% End: 

\section{Stability and Stabilization of Submanifolds}\label{sec:flow_inv}

Convergence of a nonlinear system to a submanifold of state space is a requirement that appears in many applications, including coordination of multi-agent systems \cite{tanner2007flocking}, synchronization of oscillators \cite{dorfler2014synchronization}, computational neuroscience \cite{burak2009accurate}, and nonlinear control design
\cite{astolfi2007nonlinear}. Design of controllers to stabilize submanifolds has been investigated recently using transverse feedback linearization \cite{nielsen2008local}, reduction and backstepping
\cite{el2013reduction}, and modifying controllers for drift-free systems \cite{montenbruck2015compensating}. %In this section we show that the CCM conditions can be extended to give convex criteria for stability and stabilizability of submanifolds.

Suppose a submanifold is defined by a level set of some smooth funciton $
Z(t) = \{x:z(x,t) = c\}
$
where $z:\R^n\times\R^+\to \R^q$ and $\pder[z]{x}$ has rank $q$ for all $x$. The sets $Z(t)$ are called 
{\em controlled invariant} if there exists a smooth mapping $u^\star:Z(t)\times \R^+\to \R^m$ such that 
\[\pder[z(x,t)]{t}+\pder[z(x,t)]{x}(f(x,t)+B(x,t)u^\star(x,t) )=0
\]
for all $x\in Z(t)$ for all $t\in\R^+$. In the case of an uncontrolled system, this reduces to the condition for a manifold to be {\em flow invariant}: $\pder[z]{t}+\pder[z]{x}f(x,t)=0$.

The objective is to design a controller guaranteeing exponential convergence to $Z(t)$, i.e. for each $x(0)$ one can construct a control signal $u(t)$ such that the solution $x(t)$ of \eqref{eq:sys} satisfies
$
\inf_{y\in Z(t)} |x(t)-y|\le e^{-\lambda t} \bar R
$ for some $\bar R>0$.

Assume we can construct a smooth matrix function $G(x,t)$ with columns that form a basis for the null space of $\pder[z(x,t)]{x}$. In order to study stability and stabilization of $Z(t)$, we construct a ``virtual control system'':
\begin{equation}\label{eq:virt}
\dot x = f(x,t)+\bar B(x,t)\bar u
\end{equation}
where $\bar B(x,t) = [B(x,t) \  G(x,t)]$ and $\bar u = [u' \ v']'$, with $u$ the actual control input and $v$ a newly introduced ``virtual control''.

\begin{thm}\label{thm:manifold_stab}
If there exists a CCM for the virtual control system \eqref{eq:virt} satisfying the strong conditions C\ref{C1}, C\ref{C2}, then any time-varying submanifold of the form $Z(t) = \{x:z(x,t) = c\}$ can be exponentially stabilized (open-loop, sampled-data, or continuously almost everywhere) with rate $\lambda$.
\end{thm}

The proof uses the concept of a ``shadow state'' $\bar x(t)$, which has the property that $\bar x(t) \in Z(t)\, \forall t$  and can be thought of as a generalised projection of $x(t)$ onto $Z(t)$. The virtual control system \eqref{eq:virt} is constructed so that it can represent dynamics of the real system when $v=0$, but can also represent $\bar x(t)\in Z(t)$ when $u=x^\star(\bar x,t)$ and $v$ is arbitrary. The idea is to ensure $x(t)$ converges to $\bar x(t)$, and therefore to $Z(t)$.

The following corollary gives simple convex criteria for an uncontrolled system to converge to a submanifold. 
\begin{cor}\label{thm:synch} Consider an uncontrolled system of the form \eqref{eq:sys} with $B=0 \ \forall x,t$.
Suppose there exists a uniformly bounded dual metric $W$, invariant on level sets of $z(x,t)$, satisfying:
\begin{equation}\label{eq:CCM_Z}
\pder[z]{x}\left(-\pder[W]{t}-\partial_fW + \pder[f]{x}W+W\pder[f]{x}'+2\lambda W\right)\pder[z]{x}'< 0,
\end{equation}
Then all solutions of the system converge exponentially with rate $\lambda$ to the set $Z(t)$.
\end{cor}

Note that this corollary is based solely on the {\em existence} of a universally stabilizing controller. Actual computation of a control signal is not required.
Corollary \ref{thm:synch} generalizes results on {\em partial contraction} in \cite{wang2005partial, Pham2007}. The latter showed that convergence to a {\em linear} manifold defined by $z(x)=Vx = 0$, with $V$ a constant matrix, is guaranteed by the condition that $V(\pder[f]{x}+\pder[f]{x}')V'$ is uniformly negative definite. 
A similar notion of {\em horizontal contraction} was studied in \cite{forni2014differential}. 

%%% Local Variables: 
%%% mode: latex
%%% TeX-master: "CCM_Journal"
%%% End: 
\section{Illustrative Example}\label{sec:examples}

One of the advantages of convex criteria is that it is possible to mix and match different design objectives for one controller. In \cite{andrieu2010uniting} and references therein, the problem of ``uniting'' locally optimal and globally stabilizing control was considered. This problem is non-trivial in a Lyapunov framework since the set of control Lyapunov functions for a system is non-convex, but in the CCM framework it is straightforward. Let us illustrate this with a particular example system taken from \cite{andrieu2010uniting}, with state $x = [x_1, x_2, x_3]'$ and dynamics \eqref{eq:sys} with
\begin{equation}\label{eq:Andrieu}
f(x) = \begin{bmatrix}-x_1+x_3\\x_1^2-x_2-2x_1x_3+x_3\\-x_2 \end{bmatrix}, \quad B = \bm{0\\0\\1}.
\end{equation}
%Note the unbounded nonlinearity appearing in the dynamics of the state $x_2$, which is not directly actuated. 
Note that this system is not feedback linearizable, since the matrix $[B, ad_fB, ad_f^2B]$ drops rank at the origin.

%For this system we have differential dynamics \eqref{eq:diffdyn} with
%\[
%\pder[f]{x} = \begin{bmatrix} 
%-1&0&1\\
%2x_1-2x_3&-1&1-2x_1\\
%0&-1&0
%\end{bmatrix}.
%\]
We first solve the linear quadratic regulator (LQR) problem for the system linearized at the origin with cost function $\int_0^\infty (x'x+ru^2)dt$ with $r=1$, obtaining a solution $P=P'>0$ of the algebraic Riccati equation, and the locally optimal controller $u = -r^{-1}B'Px$. Then we can search for $W$ and $\rho$ satisfying \eqref{eq:CCM_rhoform_dual} and the additional linear constraints $W(0) = P^{-1}$ and $\rho(0) = 2r^{-1}$, so that locally the LQR and CCM controllers are the same. To satisfy Condition C\ref{C2}, entries of $W$ were allowed to be quadratic functions of $x_1$ and $x_2$. The resulting metric is not uniformly-bounded, but still satisfies the  conditions of Lemma \ref{lem:hopf_rinow}. Similarly, $\rho$ was a quadratic polynomial in $x_1$, and $\lambda = 0.5$. 
As an optimization objective we chose the $l^1$ norm of the coefficients of the polynomial entries of $W$ to encourage sparsity. 
The metric was found using sum-of-squares programming \cite{parrilo2003semidefinite} via the parser Yalmip \cite{lofberg2004yalmip} and solver Mosek, the resulting semidefinite program took about 0.4 seconds to solve on a standard desktop computer. 

It can be seen in Fig. \ref{fig:A1} that for small initial conditions the CCM controller and the LQR are virtually identical. This is because the minimal geodesic is close to a straight line and $\rho$ and $W$ are almost unchanged, so the CCM control law approximates a simple linear feedback on $x-x^\star$. In contrast, for larger initial conditions the LQR controller was not stabilizing, while the CCM controller was. Simulations under LQR diverge rapidly after about 2 seconds. Further results on CCMs for poblems  in robotics can be found in \cite{CCM_ISRR}, and a method for computing geodesics can be found in \cite{leung2017}.

\begin{figure}
\begin{center}
\includegraphics[width=0.5\columnwidth]{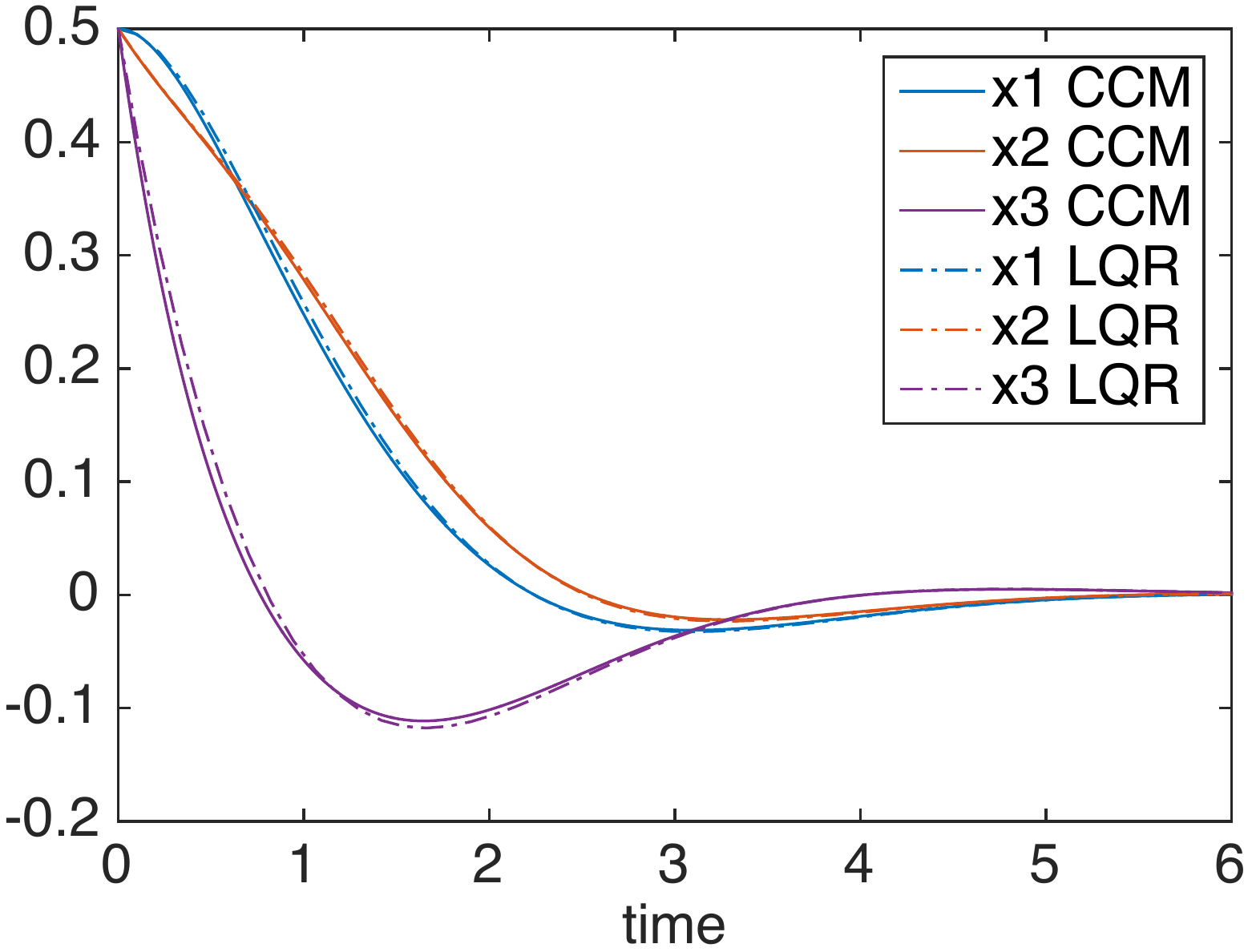}\includegraphics[width=0.5\columnwidth]{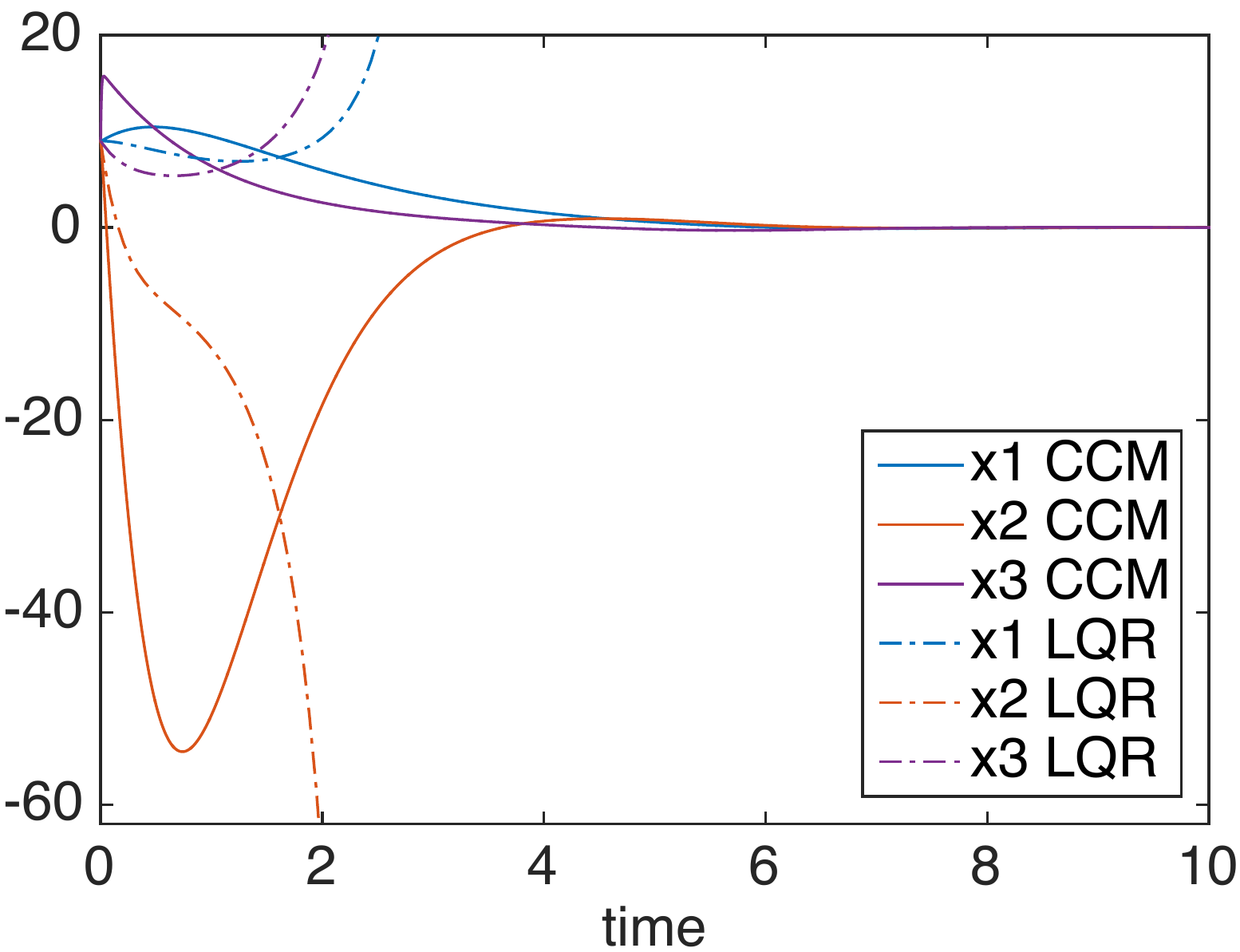}
\caption{Response of System \eqref{eq:Andrieu} with CCM and LQR control to initial state $x(0) = [0.5, 0.5, 0.5]'$ (left) and $x(0) = [9, 9, 9]'$ (right). This exhibits the ``locally optimal'' and ``globally stabilizing'' behaviour of the CCM controller.
}
\label{fig:A1}
\end{center}
\end{figure}

\appendix
\begin{proof}[Proof of Proposition 1] 
By assumption that the closed-loop system is strictly  contracting, we have $\ddt (\delta_x'M(x,t)\delta_x)< -2\lambda \delta_x'M(x,t)\delta_x$ for all $x,v,t$ and $\delta_x\ne 0$. Expanding the left hand side $\ddt (\delta_x'M(x,t)\delta_x)=
 \delta_x'\left[\pder[M]{t} +\partial_fM + \left(\pder[f]{x}+BK\right)'M+M\left(\pder[f]{x} + BK\right)\right.\notag\\
+\left.\sum_{i=1}^m (k_i(x,t)+v_i) \left(\partial_{b_i}M+\pder[b_i]{x}'M+M\pder[b_i]{x}\right)\right]\delta_x.\notag
$
Note that this is affine in $v_i$, so if this to be bounded for all $v\in\R^m$, it is clearly required that $
\partial_{b_i} M + \pder[b_i]{x}'M+M\pder[b_i]{x} = 0
$ for each $i$. Furthermore, if $\delta_x'MB=0$ then terms involving $K$ vanish and hence $\delta_x'\left(\pder[M]{t}+\partial_f M+\pder[f]{x}'M+M\pder[f]{x}+2\lambda M\right)\delta_x <0$, and the result follows from direct calculation of $A$ and $\dot M$.
\end{proof}

%%% Local Variables: 
%%% mode: latex
%%% TeX-master: "CCM_Journal"
%%% End: 

\begin{lem} \label{lem:hopf_rinow}Suppose a dual metric $W(x)=M(x)^{-1}$ satisfies a quadratic bound on its largest eigenvalue: $\lambda_{\max}(W(x)) \le |\mf Ax+\mf B|^2$ for all $x\in\R^n$, for some fixed matrices $\mf A,\mf B$, where $|\cdot|$ is the Euclidean norm. Then there exists a minimal geodesic between any pair of points in $\R^n$.
\end{lem}

\begin{proof}
	By the Hopf Rinow theorem, the result follows if any geodesic segment can be extended indefinitely \cite{docarmo1992riemannian}. 
	By assumption, the metric $M(x)$ satisfies the lower bound $\lambda_{\min}(M(x))\ge \frac{1}{|\mf Ax+\mf B|^2}$. Take any geodesic $\gamma(s)$ defined on some interval $s\in (a, b)\subset \R$. Geodesics have constant speed: $\gamma_s'M(\gamma)\gamma_s=c$ for some $c>0$, so
$
{|\gamma_s|}\le c{|\mf A\gamma+\mf B|}.
$
This implies that $|\gamma|$ grows at-worst exponentially as the parameter $s\rightarrow \pm \infty$. Therefore solutions do not exhibit finite escape in either direction, so the interval of existence is $(-\infty, \infty)$.\end{proof}

%\begin{lem}\label{lem:lVdecr}  Consider the system \eqref{eq:sys}, \eqref{eq:diffdyn}, and smooth real-valued functions $V(x,\delta_x,t)$ and $\kappa(x,\delta_x,t)$. 
%Suppose \eqref{eq:diffCLF} is satisfied, then an admissable differential feedback can be explicitly constructed.
%\end{lem}
\begin{lem}\label{lem:lVdecr} Consider the system \eqref{eq:sys}, \eqref{eq:diffdyn}, and  smooth real-valued ``differential storage function'' $V(x,\d_x,t)$, and ``differential supply rate'' $\kappa(x,\d_x,t)$. If for all $x,u,\d_x\ne 0, t$
\begin{equation}\label{eq:diffCLF}
\pder[V]{\delta_x}B = 0 \Longrightarrow  \pder[V]{t} + \pder[V]{x}(f+Bu)+\pder[V]{\delta_x}A\delta_x < \kappa.
\end{equation}
then a ``differential feedback controller'' $k_\delta(x,\delta_x,u,t)\in\R^m$ exists that satisfies the following two properties:	
%The main strategy of this paper is to prove that $x(t)$ can be made to converge to $x^\star(t)$ by showing that the length of a curve joining them can be made to shrink sufficiently fast. To this end, consider smooth real-valued functions $V(x,\d_x,t)$, a ``differential storage function'', and  $\kappa(x,\d_x,t)$, a ``differential supply rate''. Suppose there exists a differential feedback control $\delta_u(x,\d_x,u,t)$ that satisfies the following properties: 
\begin{enumerate}
	\item \textbf{Closed-loop dissipativity}: for all $x,u,\d_x\ne 0, t$
	\[
\dot V = \pder[V]{t} + \pder[V]{x}(f+Bu)+\pder[V]{\delta_x}(A\delta_x+Bk_\delta)<\kappa.
\]
	\item \textbf{Path-integrability}: for any regular curve  $c$, and any $u_0\in\R^m, t\in\R^+$, a unique solution of the following integral equation exists on $s\in[0,1]$:
\begin{equation}\label{eq:control_int}
\u(s) = u_0+\int_0^s k_\delta (c(\s),c_s(\s), \u(\s), t)d\s.
\end{equation}
\end{enumerate}
\end{lem}

\begin{proof}%[Proof of Lemma \ref{lem:lVdecr}]
%We construct a differential feedback control and integrate it along the path $\gamma$. The critical fact to establish is that the resulting integral equation does not blow up.
For brevity of notation, let us define
\begin{align}
\mf a(x,\delta_x,u,t) &:= \pder[V]{t}+\pder[V]{x}(f+Bu)+\pder[V]{\delta_x}A\delta_x -\kappa \notag\\
\mf b (x,\delta_x,t) &:= \pder[V]{\delta_x}BB'\pder[V]{\delta_x}'.\notag
\end{align}
Note that by construction $b\ge 0$, and by assumption \eqref{eq:diffCLF}, for all $x, \delta_x\ne 0, u, t$, either $\mf a<0$ or $\mf b>0$. Now, define
\begin{equation}\label{eq:rhocases}
	\rho(x,\delta_x,u,t) := \begin{cases}
    0,& \text{if } \mf a<0,\\
    \frac{\mf a+\sqrt{\mf a^2+\mf b^2}}{\mf b},       & \text{otherwise}.
    
\end{cases}
\end{equation}
It follows from \cite[Thm 1]{sontag1989universal} that $\rho$ is a smooth for all $x, \delta_x\ne 0, u, t$.
%By assumption , and under these conditions $\rho$ is a smooth function of $\mf a, \mf b$ \cite{sontag1989universal}.
%$\rho = \rho_0+\sum \rho_iu_i$. We claim that for every fixed $t$, each of $\rho_i, i = 0, 1, ..., m$ is bounded on any compact subset $S$ of $x, \delta_x$. Each is of the form $\rho_i=\mf a_i/\mf b$ if $\mf a\ge 0$ and zero otherwise. For each $t$ set  $\bar S(t) = S\cap {x, \delta_x: \mf a = 0}$ is compact. Depends on $u$.
%
%For each fixed $t$, on any compact subset of $x, \delta_x$, each of $\mf a_0,..., \mf a_m, \mf b$ is bounded, by continuity.
%
%Note that the contrapositive of condition \eqref{eq:diffCLF} states that $\mf a\ge 0 \Longrightarrow\mf  b>0$, so division by $\mf b$ in the formula for $\rho$ is not problematic.
%Also that, since $a(x,u,\delta_x,t)$ is affine in $u$ for each $x,\delta_x,t$, so is $\rho(x,u,\delta_x,t)$. 
Now we construct the control:
\begin{equation}\label{eq:diff_fb}
k_\delta(x,\delta_x,u,t) = -\rho(x,\delta_x,u,t)B(x,t)'\pder[V(x,\delta_x, t)]{\delta_x}'.
\end{equation}
Substituting into \eqref{eq:diffdyn} establishes closed-loop dissipativity:
$
\ddt V = \mf a-\rho \mf b+\kappa = \kappa -\sqrt{\mf a^2+\mf b^2}<\kappa.
$

We now prove path-integrability, i.e. that a solution of \eqref{eq:control_int} exists. We will prove this by contradiction. By assumption the curve $c$ is regular, so $c_s(\s)\ne 0$ for all $\s$, so $k_\delta$ is a smooth function of its third argument for all $\s$. Hence integrability follows unless there is finite escape at some $s=\bar s\le 1$. It is clear from \eqref{eq:diff_fb} that this would imply $\rho\rightarrow \infty$ as $s\rightarrow \bar s$.

First, we observe that $\mf b>0$ in a neighbourhood of $\bar s$, since if $\mf b=0$ then $\rho=0$, but we require $\rho$ to blow up.

Second, we note that the only dependence $\rho$ has on $u$ is via $\mf a$, which by construction is an $s$-dependent affine function of $u$. It follows from  \eqref{eq:rhocases} that if there is a closed interval $\mf S\subset [0,1]$ such that $\mf b>0$ for $s\in\mf S$, then $\rho$ is a globally Lipschitz function of $u$ on $\mf S$. By standard comparison results, e.g., \cite[Thm 3.2]{khalil2002nonlinear}, a unique solution to \eqref{eq:control_int} exists on this interval, which contradicts finite escape at  $s=\bar s$.
\end{proof}

\begin{lem}\label{lem:OL}
	Given a control contraction metric $M(x,t)$, a time interval $\mathcal T=[t_i,t_{i+1}]\subset\R^+$, and a path $c(t_i)$ connecting $x^\star(t_i)$ to $x(t_i)$, suppose the open-loop control signal in Section \ref{sec:OL_construction} is applied on $\mathcal T$ initalized with $c$, then  for all $t\in\mathcal T$
\begin{equation}\label{eq:distance_bound_OL}
d(x^\star(t),x(t),t)\le e^{-\lambda(t-t_i)}L(c(t_i),t_i).
\end{equation}
\end{lem}
\begin{proof}
By construction, $c(t,0)=x^\star(t)$ and $c(t,1)=x(t)$ for all $t\in\mathcal T$. Furthermore, $c_s(t,s)=\pder[c]{s}(t,s)$ satisfies $
\ddt c_s(t,s) = A(c(t,s),k_p,t)c_s+B(c(t,s),t)k_\delta(t,s)
$
for all $t\in \mathcal T, s\in[0,1]$. 

By construction of $k_\delta$ we have
$
\ddt (c_s'M(c,t)c_s)<-2\lambda c_s'M(c,t)c_s.
$
Integrating with respect to $s$ gives
$
\ddt E(c(t),t)<-2\lambda E(c(t),t)
$
and integrating with respect to $t$ gives
$
E(c(t),t)\le e^{-2\lambda(t-t_i)}E(c(t_i),t_i)
$
for $t\in\mathcal T$, with strict inequality for $t>t_i$. 
Taking square roots gives
$
L(c(t),t)\le e^{-\lambda(t-t_i)}L(c(t_i),t_i).
$
Now, for each $t$ the curve $c(t)$ connects $x^\star(t)$ to $x(t)$, and the Riemannian distance between these points is the infimum of lengths of such curves, so we obtain the bound \eqref{eq:distance_bound_OL}.
\end{proof}

\begin{proof}[Proof of Theorem \ref{thm:CCM}]
We will show that for each type of controller, the Riemannian distance between $x(t)$ and $x^\star(t)$ decreases exponentially:
\begin{equation}\label{eq:dist_decrease}
d(x^\star(t),x(t),t)\le e^{-\lambda t}d(x^\star(0),x(0),0)
\end{equation}
Exponential convergence in the Euclidean metric is then implied by uniform boundedness of $M(x,t)$. In particular, it is straightforward to show that $\sqrt{\alpha_1}|x(t)-x^\star(t)|\le d(x^\star(t),x(t),t)$ and $d(x(0),x^\star(0),0)\le \sqrt{\alpha_2}|x(0)-x^\star(0)|$. Combining with \eqref{eq:dist_decrease} gives:
\[
\sqrt{\alpha_1}|x(t)-x^\star(t)|\le e^{-\lambda t}\sqrt{\alpha_2}|x(0)-x^\star(0)|,
\] i.e. \eqref{eq:expstab} holds with $R= \sqrt{\frac{\alpha_2}{\alpha_1}}$.

For the open-loop control, Lemma \ref{lem:OL} implies that for any $t>0$ we have $d(x^\star(t),x(t),t)\le e^{-\lambda t}L(c(0),0)$. If the minimal geodesic is chosen for $c(0)$ then $L(c(0),0)=d(x^\star(0),x(0),0)$, and so we obtain \eqref{eq:dist_decrease}.

Similarly, for the sampled-data controller, on each interval $[t_i, t_{i+1})$ we have
\[
d(x^\star(t),x(t),t)\le L(c(t),t)\le e^{-\lambda(t-t_i)}L(c(t_i),t_i),\]
and then at time $t_{i+1}$ a minimal geodesic $\gamma_{i+1}$ is computed which has length $L(\gamma_{i+1})=d(x^\star(t_{i+1}),x(t_{i+1}),t_{i+1})\le \lim_{t\rightarrow t_{i+1}} L(c(t),t)$ where the limit in $t$ is from the left.

For continuous feedback, we first show that the given controller is smooth on $\mathfrak D(x^\star,t)$. It follows from  \cite[Prop. 3.5, p. 117]{docarmo1992riemannian} that the mapping $(x,x^\star,t)\mapsto \gamma$ is smooth. Now, it follows from Lemma \ref{lem:lVdecr} and the smoothness of $\rho$ defined in \eqref{eq:rhocases} that the mapping $(x,x^\star,t)\mapsto u$ is smooth. 

%Note also that with $V(x,\delta_x,t)=\delta_x'M(x,t)\delta_x$ and $\rho, B$ smooth functions, $k_\delta$ defined in \eqref{eq:diff_fb} goes to zero smoothly as $\delta\rightarrow 0$, hence $k_p$ is smooth at $x=x^\star$.

%%%

To show continuity at $x=x^\star$, we first note that since $V=\d_x'M\d_x$ is quadratic in $\d_x$, and the differential dynamics \eqref{eq:diffdyn} are linear in $\d_x$, the {\em small control property} of \cite{sontag1989universal} holds, i.e. for any $\epsilon>0$ there exists $\eta>0$ such that for $\d_x\ne 0$ with $|\delta_x|\le \eta$ there exists $\d_u$ with $|\delta_u|< \epsilon$ satisfying $\dot V<0$. This implies that $k_\delta$ in \eqref{eq:diff_fb} is continuous in $\delta_x$ at $\d_x=0$.

%As $x\rightarrow x^\star$ we have $\gamma(s)\rightarrow s(x-x^\star)$ since short geodesics are straight lines. So $ \delta_u \approx k_\d(x^\star, x-x^\star,u^\star,t)$.

%It follows from  that, for a sufficiently small all points $x$ in a neighbourhood of $x^\star$ lie on some geodesic $\gamma^0$, this suffices to show continuity of the controller at $x=x^\star$.

Let us show convergence of the control law to $u^\star$ for a sequence of states $x^k$ with $\lim_{k\rightarrow\infty}x^k= x^\star(t)$. It suffices to construct the sequence $x^k$ along a particular but arbitrary geodesic $\gamma^0$, e.g. $x^k=\gamma^0(1/k)$, since all states in a neighborhood of $x^\star$ lie on such a geodesic \cite[Thm 3.7]{docarmo1992riemannian}.  Let $\gamma^k$ denote the segment from $x^\star(t)$ to $x^k$. Now, the constant-speed property of geodesics states that  $\|\g_s^k\|_s=L(\g^k,t)$ for all $s$, and so   by uniform-boundedness of the metric $|\gamma_s^k|\rightarrow 0$ uniformly in $s\in[0,1]$.
%Now, in a neighbourhood of $x^\star(t)$, either $\mf a<0$ or $\mf b>0$. In the first case, $k_\delta=0$ nearby $x^\star$, so continuity is clear. In the second case, as argued for Lemma \ref{lem:lVdecr}, $k_\delta$ is uniformly Lipschitz in $u$ on some closed interval $s\in[0,\bar s]$. 
Therefore, by continuous dependence of \eqref{eq:k_p} on $k_\delta$ (e.g. \cite[Lemma 3.1]{Hale80}) the sequence of solutions $k_p(\gamma^k,u^\star(t),t,1)$ of \eqref{eq:k_p} converges to $u^\star(t)$ as $k\rightarrow \infty$.

To show exponential decrease in distance, recall \cite[Prop 2.4, p195]{docarmo1992riemannian} that on $\mathfrak D(x^\star,t)$ the energy is a smooth function of its endpoints, i.e.  $(a,b)\mapsto E(a,b,t)$ is smooth for each $t$.
Now, consider the open-loop control initialised with $t_i=t$ and $\gamma$ for $c(t_i)$. This open-loop control is identical to the proposed continuous feedback at time $t$, and therefore $\dot x(t)$ is also identical.  Hence there exists a (non-minimal) path $c(\cdot)$ defined on $[t,t+\epsilon)$ with $c(t)=\gamma_i$ satisfying
$
\ddt E(c(t),t)<-2\lambda E(\gamma,t),
$
 therefore $\ddt E(x(t),x^\star(t),t)<-2\lambda E(x(t),x^\star(t),t)$. Integrating with respect to time gives the result.
\end{proof}

\begin{proof}[Proof of Theorem \ref{thm:coord}]
For the first statement, under the feedback transformation we have a new control system affine in $v$:
$
\dot x=(f+B\alpha) + B\beta v,
$
and the associated differential dynamics are of the form \eqref{eq:diffdyn} with
$
A_v(x,v) =A(x,u(x,v))+B(x)\pder[u(x,v)]{x}$ and $B_v(x) = B(x)\beta(x).
$
Now, consider the dual CCM condition \eqref{eq:weak_kernel}. The same annihilator matrix $B_\perp$ can be used since $B_\perp'B_v = B_\perp B \beta =0$. Now, substitute $A_v$ for $A$ in \eqref{eq:weak_kernel}, and notice that since $B_\perp'B = 0$ the second term ($B\pder[u]{x}$) in $A_v$ has no effect. The first term in $A_v$  is just $A(x,u)$ evaluated a particular value of $u$, and \eqref{eq:weak_kernel} holds for all $u$, and hence can be applied under feedback transformation. 

For the second statement,
%under the mapping $x\mapsto \xi$, tangent vectors transform as $\delta_\xi = \Phi(x,t)\delta_x$, and 
the differential dynamics transforms as
$
\dot \delta_\xi=A_\xi\delta_\xi + {B_\xi}\delta_u, 
$
with $A_\xi=\dot \Phi\Psi+\Phi A\Psi$ and $B_\xi:=\Phi B$. 
Taking $M_\xi = \Psi' M \Psi$ and the identity $\dot\Psi = -\Psi\dot\Phi\Psi$ we have
$
\dot M_\xi =  -\Psi'\dot\Phi'\Psi' M\Psi + \Psi' \dot M \Psi - \Psi' M \Psi\dot\Phi\Psi,
$
then straightforward calculation gives
$
\dot M_\xi +A_\xi'M_\xi + M_\xi A_\xi + 2\lambda M_\xi = \Psi'(\dot M + A'M + MA +2\lambda M)\Psi
$ and the result follows from the fact that $\delta_x=\Psi \delta_\xi$. The third statement then follows by considering points $x=\phi^{-1}(\xi)$.
\end{proof}

\begin{proof}[Proof of Theorem \ref{thm:manifold_stab}] We will construct the open-loop controller, but others are analogous to the construction for Theorem \ref{thm:CCM}. 
Since there exists a CCM $M(x,t)$ satisfying the strong conditions, we can construct a differential feedback controller for \eqref{eq:virt} of the form
$
k_\delta(x,\delta_x,t) = -\frac{1}{2}\rho(x,t)\bar B(x,t)'M(x,t)\delta_x
$. By construction of $\bar B$ this decomposes as
$
\delta_u  = -\frac{1}{2}\rho B'M\delta_x, \ \delta_v  = -\frac{1}{2}\rho G'M\delta_x.
$
Now, for a given path $c(t)\in\Gamma(\bar x(t), x(t))$ we construct the following paths of signals for the real and virtual control inputs:
\begin{align}\label{eq:control_int_v1}
u(t,s) =& u^\star(\bar x,t)+\int_0^s \delta_u (c(\s),c_s(\s), t)d\s,\\
v(t,s) =& -\int_s^1 \delta_v (c(\s),c_s(\s), t)d\s,
\end{align}
noting that $\pder[u]{s}=\delta_u$ and $\pder[v]{s}=\delta_v$. Therefore, by an analogous argument to Theorem \ref{thm:CCM}, one can construct open-loop, sampled-data, or continuous almost-everywhere controllers such that the length of $c(t)$ shrinks exponentially.

%
%Now, given $\bar x(t)\in Z(t)$ and any smooth path $c(t)\in\Gamma(\bar x(t), x(t)$ we construct smooth paths of controls: $
%k_p(t,s) = u^\star(\bar x,t)+\int_0^s \delta_u (c(\s),c_s(\s), t)d\s$ and $
%k_v(t,s) = -\int_s^1 \delta_v (c(\s),c_s(\s), t)d\s.
%$
%I.e. for each $s\in[0,1]$ solve \eqref{eq:virt} with $x(t)=c(t,s), u(t)=k_p(t,s), v(t)=k_v(t,s)$. It follows from the same reasoning as Lemma \ref{lem:OL} that the length of $c(t)$ shrinks exponentially.

Now, notice that at $s=1$ we have $v(t,1)=0$, so $\dot c(t,1) = f(c(t,1),t)+B(c(t,1),t)u(t)$ and so $c(t,1)=x(t)$ for all $t$, i.e. the ``real'' dynamics are recovered.
On the other hand, at $s=0$ we have the $c(t,0)=\bar x(t)$ (the ``shadow'' state) with dynamics $
\dot {\bar x} = f(\bar x,t)+B(\bar x,t)u^\star(\bar x,t)+G(\bar x,t)k_v(t,0).
$
But by the assumption that $Z$ is invariant for $\dot x=f+Bu^\star$, and that $\pder[z]{x}G=0$, it follows that $\bar x(t)\in Z(t)$ for all $t$.

So, since $L(c(t))\rightarrow 0$ exponentially and $c(t,1)=x(t)$ and $c(t,0)\in Z(t)$, it follows that $x(t)\rightarrow Z(t)$ exponentially.
\end{proof}

\bibliographystyle{IEEEtran}
\bibliography{elib}

\end{document}